\documentclass[3p]{elsarticle}

\journal{Discrete Applied Mathematics}

\usepackage[utf8]{inputenc}
\usepackage{amsmath,amssymb,amsthm}
\usepackage{graphicx}
\usepackage{tikz}
\usepackage{subfigure}
\usepackage{myalgorithm,myalgorithmic}

\newtheorem{theorem}{Theorem}

\newtheorem{lemma}[theorem]{Lemma}
\newtheorem{observation}[theorem]{Observation}
\newtheorem{definition}[theorem]{Definition}
\newtheorem{corollary}[theorem]{Corollary}
\newtheorem{conjecture}[theorem]{Conjecture}
\newtheorem{proposition}[theorem]{Proposition}

\newcommand{\M}{\gamma^{\text{\tiny{ID}}}}
\newcommand{\OM}{\gamma^{\text{\tiny{OID}}}}

\begin{document}
\begin{frontmatter}
\title{On the size of identifying codes in triangle-free graphs\tnoteref{t1}}
\tnotetext[t1]{This research is supported by the ANR Project IDEA 
- Identifying coDes in Evolving grAphs, ANR-08-EMER-007, 2009-2011 
and by the KBN Grant 4~T11C~047~25.}
\author[LaBRI]{Florent Foucaud\corref{cor}}
\author[LaBRI]{Ralf Klasing}
\author[LaBRI,Gdansk]{Adrian Kosowski}
\author[LaBRI]{Andr\'e Raspaud}
\address[LaBRI]{Univ. Bordeaux, LaBRI, UMR5800, F-33400 Talence, France.\\
CNRS, LaBRI, UMR5800, F-33400 Talence, France.\\
INRIA, F-33400 Talence, France.}
\address[Gdansk]{Department of Algorithms and System Modeling, Gda\'nsk University of Technology, Narutowicza 11/12, 80952 Gda\'nsk, Poland.}
\cortext[cor]{Corresponding author. E-mail: foucaud@labri.fr - Telephone: +33540003517 - Fax: +33540006669}

\begin{abstract}
In an undirected graph $G$, a subset $C\subseteq V(G)$ such that $C$
is a dominating set of $G$, and each vertex in $V(G)$ is dominated by
a distinct subset of vertices from $C$, is called an identifying code
of $G$. The concept of identifying codes was introduced by Karpovsky,
Chakrabarty and Levitin in 1998. For a given identifiable graph $G$,
let $\M(G)$ be the minimum cardinality of an identifying code in
$G$. In this paper, we show that for any connected identifiable
triangle-free graph $G$ on $n$ vertices having maximum degree
$\Delta\geq 3$, $\M(G)\le n-\tfrac{n}{\Delta+o(\Delta)}$. This bound
is asymptotically tight up to constants due to various classes of
graphs including $(\Delta-1)$-ary trees, which are known to have their
minimum identifying code of size $n-\tfrac{n}{\Delta-1+o(1)}$. We also
provide improved bounds for restricted subfamilies of triangle-free
graphs, and conjecture that there exists some constant $c$ such that
the bound $\M(G)\le n-\tfrac{n}{\Delta}+c$ holds for any nontrivial
connected identifiable graph $G$.
\end{abstract}


\end{frontmatter}

\section{Introduction}

Identifying codes, which have been introduced in~\cite{KCL98}, are
dominating sets having the additional property that each vertex of the
graph can be uniquely identified using its neighbourhood within the
identifying code. They have found numerous applications, such as
fault-diagnosis in multiprocessor networks~\cite{KCL98}, the placement
of networked fire detectors in complexes of rooms and
corridors~\cite{DRSTU03}, compact routing~\cite{CLST07}, or the
analysis of secondary RNA structures~\cite{HKSZ06}. Identifying codes
are a variation on the earlier concept of locating-dominating sets
(cf.~e.g.~\cite{CSS87,S87,RS84}), and a special case of the more
general test cover problem~\cite{DHHHLRS03,MS85}. Identifying codes
have been studied in specific graph classes such as
cycles~\cite{BCHL04,GMS06}, trees~\cite{BCHL05,CGHLMM06},
grids~\cite{KCL98} or hypercubes~\cite{JL09,RHL02}. Extremal problems
regarding the minimum size of an identifying code have been studied
in~\cite{CHL07,F09,FGKNPV10,FNP10,GM07,M06}.

Herein, we further investigate these extremal questions by giving new
upper bounds on the size of minimum identifying codes for
triangle-free graphs using their maximum degree.

\subsection{Notations and definitions}
Let $G=(V,E)$ be a simple undirected graph. We denote the vertex set
of $G$ by $V=V(G)$ and its edge set by $E=E(G)$. We also denote by
$n=|V|$ the order of $G$ and by $\Delta=\Delta(G)$ the maximum vertex
degree of $G$.

For a vertex $v$ of $G$, the \emph{ball} $B(v)$ is the set of all
vertices of $V$ which are at distance at most 1 from $v$. We
denote by $N(v)=B(v)\setminus\{v\}$, the \emph{neighbourhood} of
$v$. For a set $X$ of vertices of $G$, we define $N(X)$ to be the
union of the neighbourhoods of all vertices of $X$, that is
$N(X)=\cup_{x\in X}N(x)$. Whenever we find it necessary to emphasize
on the graph $G$ for which the neighbourhood is considered, we write
$B_G(u)$, $N_G(u)$ and $N_G(X)$. Two distinct vertices $u,v$ are
called \emph{twins} if $B(u)=B(v)$~\cite{CHHL07}. They are called
\emph{false twins} if $N(u)=N(v)$ but $u$ and $v$ are not
adjacent~\cite{BU82}.

For a subset $S$ of vertices of $G$, we denote by $G[S]$ the subgraph
of $G$ induced by $S$. A \emph{matching} $M$ of a graph $G$ is a
subset of edges of $G$ such that no two edges of $M$ have a common
vertex. If within the set of all endpoints of the edges of $M$ no
other edges than the ones of $M$ exist, we call $M$ an \emph{induced
  matching}.

Given a set $S$ of vertices of $G$, we say that a vertex $x$ of $G$ is
\emph{$S$-isolated} if $x\in S$ and no neighbour of $x$ belongs to
$S$. 
We say that vertex $u$ \emph{dominates} vertex $v$ if $v\in B(u)$. For
two subsets $C,U$ of vertices,  $C$ \emph{dominates} $U$ if
each vertex of $U$ is dominated by some vertex of $C$. Set $C\subseteq
V$ is called a \emph{dominating set} of $G$ if $C$ dominates $V$. The
vertices of a pair $u,v$ of vertices of $V$ are \emph{separated}
by some vertex $x\in V$ if $x$ dominates exactly one of the vertices
$u$ and $v$. We call $C\subseteq V$ an \emph{identifying code} of $G$
if it is a dominating set of $G$, and for all pairs $u,v$ of vertices
of $V$, $u$ and $v$ are separated by some vertex of $C$. The
latter condition can be equivalently stated as $B(u)\cap C\neq
B(v)\cap C$, or as $(B(u)\oplus B(v))\cap C\neq\emptyset$ (denoting by
$\oplus$ the symmetric difference of sets). In the following, we might
simply call an identifying code a \emph{code} and a vertex of the
code, a \emph{code vertex}. Given a graph $G$ and a subset $S$ of its
vertices, we say that a set $C\subseteq S$ is an \emph{$S$-identifying
  code} of $G$ if $C$ is an identifying code of $G[S]$.

A graph is said to be \emph{identifiable} if it admits an identifying
code. This is the case if and only if it does not contain any pair of
twins~\cite{KCL98}. An example of a graph which is not identifiable is
the complete graph $K_n$. For an identifiable graph $G$, we denote by
$\M(G)$ the cardinality of a minimum identifying code of $G$. The
problem of determining the exact value of $\M(G)$ is known to be an
NP-hard problem, even when $G$ belongs to the class of planar graphs
of maximum degree~4 having arbitrarily large girth~\cite{A10}, or to
the class of planar graphs of maximum degree~3 and
girth~9~\cite{ACHL10}.

\subsection{Main conjecture and motivation}

This paper deals with the study of paramater $\M$ and its relation
with the order and the maximum degree of graphs. This work is an
extension of earlier results.

For any graph $G$ on $n$ vertices, the lower bound
$\M(G)\geq\lceil\log_2(n+1)\rceil$ was given in~\cite{KCL98}. This
bound is tight, and all graphs reaching it have been described
in~\cite{M06}. In~\cite{KCL98}, it was also shown that the bound
$\M(G)\ge\tfrac{2n}{\Delta+2}$ holds, and all graphs reaching this
bound have been described in~\cite{F09}. This bound is an improvement
over the $\lceil\log_2(n+1)\rceil$-bound whenever
$\Delta\leq\tfrac{2n}{\lceil\log_2(n+1)\rceil}-2$, and shows that the
maximum degree has a strong influence on the minimum possible value of
$\M$.

Considering upper bounds in terms of $n$ and $\Delta$, we conjecture
that the following bound on $\M$ holds.

\begin{conjecture}\label{conj}
  There exists a constant $c$ such that for any nontrivial connected
  identifiable graph $G$ of maximum degree~$\Delta$, $\M(G)\leq
  n-\frac{n}{\Delta}+c$.
\end{conjecture}

It is known that there exist examples of specific families of graphs
such that $\M(G)=n-\tfrac{n}{\Delta}$ (e.g. the complete bipartite
graph $K_{\Delta,\Delta}$, Sierpi\'nski graphs~\cite{GKM10} and other
classes of graphs described in the first author's master
thesis~\cite{F09}). Other classes of graphs with slightly smaller
values of parameter $\M$ are known, including graphs having high
girth. For instance, it is shown in~\cite{BCHL05} that $\M(T)=\lceil
n-\tfrac{n}{\Delta-1+1/\Delta}\rceil$ for any complete
$(\Delta-1)$-ary tree $T$ on $n$ vertices.

For all identifiable graphs having at least one edge, the upper bound
$\M(G)\leq n-1$ holds~\cite{CHL07,GM07}. This bound is tight, in
particular for the star $K_{1,n-1}$ and other graphs which have been
recently classified in~\cite{FGKNPV10}. Hence, for graphs of very high
maximum degree (say~$\Delta=n-1$), the conjecture holds since
$n-1=n-\tfrac{n}{\Delta}+\tfrac{1}{n-1}$.

Moreover, for any connected graph $G$ of maximum degree~$2$ (i.e. when
$G$ is either a path or a cycle), the exact value of $\M(G)$ is known
(see~\cite{BCHL04,GMS06}). In this case, the bound $\M(G)\leq
\tfrac{n}{2}+\tfrac{3}{2}=n-\tfrac{n}{2}+\tfrac{3}{2}$ holds and is
reached for infinitely many values of $n$ (more precisely, this is the
case when $G$ is a cycle of odd order $n\ge 7$). Hence, the conjecture
holds for $\Delta=2$.

There is some evidence that even the case $\Delta=3$ might be
challenging. Indeed, the similar notion of \emph{identifying open
  codes} (that is, identifying codes on \emph{open balls} rather than
closed balls, i.e. vertices do not dominate or identify themselves)
was studied very recently in~\cite{HY12} for cubic graphs. Denoting
$\OM(G)$ the minimum size of an identifying open code of a graph $G$,
they are able to prove that in a cubic graph $G$ admitting an
identifying open code, $\OM(G)\leq\tfrac{3n}{4}$. Moreover, they
conjecture that the only (connected) examples reaching the bound
belong to a set of six graphs, and that otherwise,
$\OM(G)\leq\tfrac{3n}{5}$, which, if true, would be sharp. This result
is proved by using a strong connection to \emph{distinguishing
  transversals} of 3-uniform hypergraphs. It is worth noting that
using the same technique in the case of (classic) identifying codes in
cubic graphs would require to handle distinguishing transversals of
4-uniform hypergraphs, which seems to be a much more difficult task.

It was shown in~\cite{FGKNPV10} that for any connected identifiable
graph $G$ of maximum degree~$\Delta$, $\M(G)\leq
n-\tfrac{n}{\Theta(\Delta^5)}$, and if $G$ is $\Delta$-regular,
$\M(G)\leq n-\tfrac{n}{\Theta(\Delta^3)}$. In this paper, we improve
these results by showing that the conjectured bound holds
asymptotically when $G$ is triangle-free. More precisely, it is proved
in Theorem~\ref{th:upperbound_notriangle-new} that $\M(G) \leq n -
\tfrac n {\Delta+o(\Delta)}$ when $G$ is a nontrivial connected
identifiable triangle-free graph. This result strongly supports
Conjecture~\ref{conj}. Moreover, the proof is constructive and can be
used to build the corresponding code in polynomial time. For some
specific subclasses of triangle-free graphs, we are able to show
bounds of the form $\M(G) \leq n - \tfrac n {\Delta+k}$ for some
constants~$k$.

\subsection{Organization of the paper}

In Section~\ref{sec:proofideas}, we begin by giving an informal
overview of the technique and the construction used to prove our
results. In Sections~\ref{sec:prelim} to~\ref{subsec:LR}, we introduce
some definitions and preliminary results that are needed in the proof
of our main result. This result is proved in
Section~\ref{subsec:mainresult}. In Section~\ref{sec:families}, we
give improved bounds for restricted subfamilies of triangle-free
graphs. Finally, Section~\ref{sec:rem} concludes the paper with a
remark on the algorithmic consequences of our proof technique.

\section{The upper bound}

\subsection{Proof ideas}\label{sec:proofideas}

The general idea of our proof technique is to construct a sufficiently
large independent set of the graph such that some specific conditions
hold. Taking the complement of this set and performing some local
modifications yields an identifying code. This technique originates
from the following proposition, which is to give the reader a first
intuition of our technique.

\begin{proposition}\label{prop:IScode}
Let $G$ be an identifiable (not necessarily connected) triangle-free
graph, and $S$, an independent set of $G$. Then, if the following
conditions hold, $V(G)\setminus S$ is an identifying code of $G$.
\begin{enumerate}
\item $S$ contains no isolated vertex of $G$.
\item For any pair $u,v$ of vertices of $S$, $N(u)\neq N(v)$ (i.e. $S$
  does not contain any pair of false twins).
\item For each vertex $v$ of degree~1 in $G$, some vertex at
  distance~2 from $v$ does not belong to $S$.
\item The graph $G[V(G)\setminus S]$ has no isolated edges.
\end{enumerate}
\end{proposition}
\begin{proof}
Let $C=V(G)\setminus S$. Since $S$ is an independent set and does not
contain any isolated vertex, $C$ is a dominating set. Let us now check
the separation condition. Let $u,v$ be an arbitrary pair of vertices
of $V(G)$. We distinguish several cases.

If $u$ and $v$ are adjacent and both have degree at least~2, since
they cannot form an isolated edge in $G[C]$, a neighbour of either one
of $u,v$ belongs to $C$ and separates them.

If $u,v$ are adjacent and one of them, say $u$, has
degree~1, since $G$ is identifiable, $v$ has at least one neighbour.
Then, by the third property of $S$, there is a vertex at distance~2 of
$u$ in $C$, separating $u$ and $v$.

If $u$ and $v$ are false twins, they do not both belong to $S$ and
hence they are separated by themselves.

Finally, if $u$ and $v$ are not adjacent and are not false twins, if
either $u$ or $v$ belong to $C$, they are separated. If both $u$ and
$v$ belong to $S$, all their neighbours belong to $C$, and since they
have distinct sets of neighbours they are separated.
\end{proof}

In order to prove our main result, we show how to build (large enough)
independent sets in triangle-free graphs such that the three first
conditions of Proposition~\ref{prop:IScode} hold (see
Lemma~\ref{lemm:goodIS}). However, it seems difficult to also ensure
that the last condition holds while keeping the size of $S$ reasonably
large. Therefore, after building $S$, we compute the set $M$ of
isolated edges of $G[V\setminus S]$ and partition $V(G)$ into the
end-vertices of $M$ (set $R$) together with their neighbours (set $L$)
on the one hand, and the remaining vertices, $V\setminus(L\cup R)$,
on the other hand. We then build a sufficiently small
\emph{$(L,R)$-quasi-identifying code} $C_1$, a variation of an
identifying code which will be defined later (see
Definition~\ref{def:quasi}). This construction is done in
Lemmas~\ref{lemma:LRdeg2} and~\ref{lemma:LR}. Setting $C_2$ as the
complement of $S$ restricted to $V\setminus(L\cup R)$, our final code
is $C_1\cup C_2$. We also combine this method with another technique
(Proposition~\ref{lemm:falsetwincode}) which is suitable for the special case
where the graph has a large number of false twins. The whole procedure
is sketched in Algorithm~\ref{alg:global_alg}.

\begin{algorithm}
  \caption{Construction of an identifying code}
  \label{alg:global_alg}
  \begin{algorithmic}[1]
    \REQUIRE{a nontrivial connected identifiable triangle-free graph $G=(V,E)$}
    \STATE Compute the set $X$ of vertices having at least one false twin
    \IF{$X$ is ``small''}
    \STATE Use Lemma~\ref{lemm:goodIS} to compute an independent set $S$ of $G$ fulfilling the three first properties listed in Proposition~\ref{prop:IScode}.
    \STATE Compute the set $R\subseteq V$ of vertices such that for each $v\in R$, $v$ has a neighbour $u$ where both $u$ and $v$ are of degree at least~2, and all the vertices of $N(u)\cup N(v)\setminus\{u,v\}$ belong to $S$.
    \STATE $L\leftarrow N(R)\setminus R$
    \STATE Compute an $(L,R)$-quasi-identifying code $C_1$ of $G$ using the constructions of Lemmas~\ref{lemma:LRdeg2} and~\ref{lemma:LR}.
    \STATE $C_2\leftarrow \left(V\setminus (L\cup R)\right)\setminus S$
    \STATE $C\leftarrow C_1\cup C_2$
    \ELSE[i.e. $X$ is ``big'']
    \STATE $C\leftarrow$ an identifying code of $G$ computed using Proposition~\ref{lemm:falsetwincode}.
    \ENDIF
    \RETURN $C$
  \end{algorithmic}
\end{algorithm}

This process is detailed in Subsection~\ref{subsec:mainresult}
(Theorem~\ref{th:upperbound_notriangle-new}). All auxiliary results
needed for this proof are developed in the next subsections.

\subsection{Preliminary results}\label{sec:prelim}

The next proposition shows how to build an identifying code of a graph
$G$ which has relatively small size when $G$ contains a large number
of false twins. We let $\equiv$ denote the \emph{false twin relation}
over $V(G)$, where $u\equiv v$ if $u,v$ are false twins. This relation
is an equivalence relation. We call an equivalence class of $\equiv$
\emph{nontrivial} if it has at least two elements.

\begin{proposition}\label{lemm:falsetwincode}
Let $G$ be a nontrivial connected identifiable triangle-free graph on
$n$ vertices and maximum degree~$\Delta$ non isomorphic to $C_4$. Let
$\mathcal{F}=\{F_1,\ldots,F_{|\mathcal{F}|}\}$ be the set of all
nontrivial equivalence classes over $\equiv$ in $G$. Then $G$ has an
identifying code of size at most $n-|\mathcal{F}|$.
\end{proposition}
\begin{proof}
First, we may suppose that $G$ is not isomorphic to $P_3$ since in
that case the lemma holds: $P_3$ has its minimum identifying
code of size~2 and $|\mathcal{F}|=1$.

For each $F_i\in\mathcal{F}$, $1\leq i\leq |\mathcal{F}|$, let $x_i$
be an arbitrary vertex of $F_i$, and let
$X=\cup_{i=1}^{|\mathcal{F}|}x_i$. We claim that if $G$ is not
isomorphic to $P_3$ or $C_4$, $C=V(G)\setminus X$ is an identifying
code of $G$. First, observe that $C$ is a dominating set of $G$. Now,
consider two vertices $x,y$. We need to show that they are separated
from each other.

If $x,y$ are false twins, the one belonging to the code separates
them. Otherwise, since $G$ is identifiable, there is a vertex $z$
which is able to separate them, say $z$ belongs to $B(x)$, but not to
$B(y)$. If $z$ belongs to the code, we are done. Otherwise, $z\in
X$.

If $z$ is a neighbour of $x$, consider a false twin
$z'$ of $z$. If $z'\neq y$, $z'$ belongs to the code and separates
$x,y$, so we are done. Otherwise, since $G$ is not isomorphic to $P_3$
and $z,y$ are false twins, one of $x$ or $y$ has another neighbour,
say $t$. If $t$ belongs to the code we are done. Otherwise, if $t$ is
a neighbour of $y$, since $G$ is not isomorphic to $C_4$, either $x$
or $y$ has another neighbour. We can repeat the argument but this
time, either this neighbour or its false twin separates $x,y$. If $t$
is a neighbour of $x$, $t$ cannot be a false twin of $y$ and therefore
either $t$ or its false twin separates $x,y$.

Finally, if $z=x$, $x$ and $y$ are not adjacent. But since they are
not false twins, there is another vertex, say $u$, with
$u\not\in\{x,y\}$, such that $u$ is adjacent to exactly one of
$x,y$. Now, either $u$ belongs to the code and we are done, or a false
twin of $u$ (which also is adjacent to exactly one of $x,y$), which
completes the proof.
\end{proof}

In the proof of our main result, we first construct an independent set
$S$ having some given properties. Then, we consider the set
$V(G)\setminus S$ as a potential code, and modify it in order to
identify those vertices which form isolated edges in $G[V(G)\setminus
  S]$. The following definition introduces a notion which helps to
formalize this situation.

\begin{definition}\label{def:LRmatch}
  Given a graph $G$ together with an induced matching $M$ of $G$, we
  denote by $R(M)$, the set of end-vertices of the edges of $M$ and by
  $L(M)$, the set of neighbours of the vertices of $R(M)$:
  $L(M)=N(R(M))\setminus R(M)$. $M$ is called a \emph{strong induced
    matching} if the following holds:
  \begin{itemize}
  \item $L(M)$ is an independent set in $G$.
  \item Each vertex $x$ of $R(M)$ has degree at least~2 in $G$ (i.e. $N(x)\cap L(M)\neq\emptyset$).
  \end{itemize}
\end{definition}

An illustration of a strong induced matching is given in
Figure~\ref{fig:LRM}. Note that in some graphs, one cannot necessarily
find a strong induced matching. Indeed, if $G$ is triangle-free, each
edge of such a matching must belong to at least some induced path on four
vertices.

\begin{figure}
\centering
\scalebox{0.5}{\begin{tikzpicture}
\path (0,0) node[draw,shape=circle] (l21) {};
\path (0,1) node[draw,shape=circle] (l22) {};
\path (0,2) node[draw,shape=circle] (l23) {};
\path (0,3) node[draw,shape=circle] (l24) {};
\path (0,4) node[draw,shape=circle] (l25) {};
\path (0,5) node[draw,shape=circle] (l26) {};
\path (0,6) node[draw,shape=circle] (l27) {};
\path (0,7) node[draw,shape=circle] (l28) {};
\path (0,8) node[draw,shape=circle] (l11) {};
\path (0,9) node[draw,shape=circle] (l12) {};
\path (0,10) node[draw,shape=circle] (l13) {};
\path (0,11) node[draw,shape=circle] (l14) {};
\path (0,12) node[draw,shape=circle] (l15) {};
\path (3,1) node[draw,shape=circle] (r21) {};
\path (3,2) node[draw,shape=circle] (r22) {};
\path (3,3) node[draw,shape=circle] (r23) {};
\path (3,4) node[draw,shape=circle] (r24) {};
\path (3,5) node[draw,shape=circle] (r25) {};
\path (3,6) node[draw,shape=circle] (r26) {};
\path (3,7) node[draw,shape=circle] (r11) {};
\path (3,8) node[draw,shape=circle] (r12) {};
\path (3,9) node[draw,shape=circle] (r13) {};
\path (3,10) node[draw,shape=circle] (r14) {};
\path (3,11) node[draw,shape=circle] (r15) {};
\path (3,12) node[draw,shape=circle] (r16) {};
\draw (l21) -- (r21)
      (l22) -- (r21)
      (l23) -- (r22)
      (l24) -- (r22)
      (l24) -- (r23)
      (l25) -- (r23)
      (l26) -- (r24)
      (l27) -- (r24)
      (l28) -- (r24)
      (l11) -- (r25)
      (l12) -- (r26)
      (l13) -- (r26)
      (l11) -- (r11)
      (l12) -- (r13)
      (l13) -- (r12)
      (l13) -- (r14)
      (l14) -- (r15)
      (l15) -- (r16);
\draw[dashed] (l23) -- +(-2,0);
\draw[dashed] (l21) -- +(-2,0);
\draw[dashed] (l23) -- +(-2,0.5);
\draw[dashed] (l23) -- +(-2,-0.5);
\draw[dashed] (l25) -- +(-2,0);
\draw[dashed] (l28) -- +(-2,0);
\draw[dashed] (l13) -- +(-2,-0.5);
\draw[dashed] (l13) -- +(-2,0.5);
\draw[dashed] (l11) -- +(-2,0);
\draw[line width=2pt]
      (r21) -- (r22)
      (r23) -- (r24)
      (r25) -- (r26)
      (r11) -- (r12)
      (r13) -- (r14)
      (r15) -- (r16);
\draw (0,5.8) ellipse (1.3cm and 7.0cm)
      (3,6.3) ellipse (1.3cm and 6.5cm)
      (-0.8,6) node {$L(M)$}
      (3.8,6.5) node {$R(M)$};
\end{tikzpicture}}
\caption{Example of a strong induced matching $M$ (thick edges) in a triangle-free graph}
\label{fig:LRM}
\end{figure}
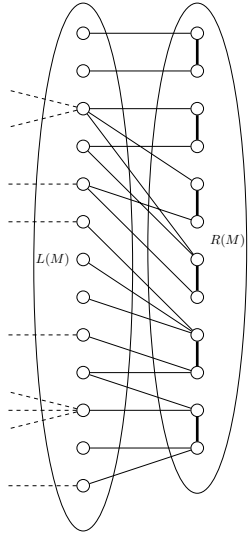

Note that in any triangle-free graph $G$ having a strong induced
matching $M$, $G[L(M)\cup R(M)]$ has no isolated edge (i.e. two
adjacent vertices of degree~1). Since in a triangle-free graph, a pair
of twins necessarily forms an isolated edge, the following observation
is immediate.

\begin{observation}\label{obs:LRM-twinfree}
  Let $G$ be a triangle-free graph having a strong induced
  matching $M$. Then $G[L(M)\cup R(M)]$ is identifiable.
\end{observation}

In order to construct small identifying codes of a triangle-free graph
$G$ having a strong induced matching $M$, we will construct special
codes for the subgraph of $G$ induced by set $L(M)\cup R(M)$. These
codes are defined as follows.

\begin{definition}\label{def:quasi}
  Let $G$ be a triangle-free identifiable graph having a strong induced matching
  $M$ with $L=L(M)$ and $R=R(M)$. Let $G'=G[L\cup R]$. We say that $C\subseteq L\cup R$ is
  an \emph{$(L,R)$-quasi-identifying code} of $G$ if:
\begin{enumerate}
\item Each vertex of $L\cup R$ is dominated by some vertex of $C$.\label{def:quasi-1}
\item For each pair $u,v$ of vertices in $L\cup R$, $C\cap
  B_{G'}(u)\neq C\cap B_{G'}(v)$, unless $u$ and $v$ both belong to
  $L$ and $N_{G'}(u)=N_{G'}(v)$.\label{def:quasi-2}
\item For each edge $e$ of $M$, at least one of the vertices of $e$
  belongs to $C$.\label{def:quasi-3}
\end{enumerate}
\end{definition}

Note that because of condition number~\ref{def:quasi-2} of 
Definition~\ref{def:quasi}, an $(L,R)$-quasi-identifying code of $G$ is not
necessarily an $(L\cup R)$-identifying code of $G$. Conversely,
because of condition number~\ref{def:quasi-3}, an $(L\cup R)$-identifying
code of $G$ might not be an $(L,R)$-quasi-identifying code of $G$.

The following proposition shows that we can use an
$(L,R)$-quasi-identifying code of $G$ to construct a valid identifying
code of $G$.

\begin{proposition}\label{prop:well-id}
  Let $G=(V,E)$ be an identifiable triangle-free graph having a strong
  induced matching $M$, with $L=L(M)$ and $R=R(M)$, and suppose that
  $L$ does not contain any pair of false twins in $G$. Also suppose that there exists
  an $(L,R)$-quasi-identifying code $C_1$ of $G$ without
  $C_1$-isolated vertices and a $(V\setminus(L\cup R))$-identifying
  code $C_2$ of $G$ where all the neighbours of vertices of $L$ within
  $V\setminus(L\cup R)$ belong to $C_2$.\footnote{Note that if a
    $(V\setminus(L\cup R))$-identifying code $C$ exists
    (i.e. $G[V\setminus(L\cup R)]$ is identifiable), then adding all
    neighbours of vertices of $L$ to $C$ yields an identifying
    code. In fact, any superset of an identifying code is still an
    identifying code.} Then, $C_1\cup C_2$ is an identifying code of
  $G$.
\end{proposition}
\begin{proof}
  We show that each pair of vertices of $G$ is separated. Since $C_2$
  is a $(V\setminus(L\cup R))$-identifying code, all pairs of vertices
  of $V\setminus(L\cup R)$ are separated. Since $C_1$ is
  $(L,R)$-quasi-identifying and there are no $C_1$-isolated vertices,
  each vertex $x$ of $L\cup R$ is dominated by at least one vertex of
  $R\cap C_1$ (see points number~\ref{def:quasi-1}
  and~\ref{def:quasi-3} of Definition~\ref{def:quasi}), which we
  denote $f_{C_1}(x)$. Moreover, by definition of sets $L$ and $R$, no
  vertex of $V\setminus(L\cup R)$ is dominated by a vertex of
  $R$. Therefore, all pairs of vertices $x,y$ with $x\in L\cup R$ and
  $y\in V\setminus(L\cup R)$ are separated by $f_{C_1}(x)$. It remains
  to check the pairs of vertices of $L\cup R$. By contradiction,
  suppose there are two vertices $u,v$ of $L\cup R$ which are not
  separated. By point number~\ref{def:quasi-2} of
  Definition~\ref{def:quasi}, $u$ and $v$ belong to $L$ and have the
  same neighbourhood within $L\cup R$. But since we assumed that they
  are not false twins and all their neighbours in $V\setminus(L\cup
  R)$ are in $C_2$, $u$ and $v$ are separated by the neighbours they
  do not have in common, a contradiction.
\end{proof}

\subsection{Building large independent sets in triangle-free graphs}\label{subsec:IS}

In order to use Proposition~\ref{prop:IScode}, we need to build (large
enough) independent sets in triangle-free graphs. We use the following
result of J.~Shearer~\cite{S83} to show that triangle-free graphs have
large independent sets which fulfill some useful conditions. Note that
the proof of the following theorem is constructive.

\begin{theorem}[\cite{S83}]\label{thm:shearer}
Let $G$ be a triangle-free graph on $n$ vertices and average
degree~$\overline{d}$. Then $G$ has an independent set of size at
least $\tfrac{\overline{d}(\ln
  \overline{d}-1)+1}{(\overline{d}-1)^2}n$.
\end{theorem}

The following corollary of Theorem~\ref{thm:shearer} is an approximate
bound which is easier to deal with and which is tight enough for our
purposes. It follows from the facts that
$\overline{d}(G)\leq\Delta(G)$ and that when $x>1$, the
function $\tfrac{x(\ln x-1)+1}{(x-1)^2}$ is decreasing. Moreover in
that case, $\tfrac{x(\ln x-1)+1}{(x-1)^2}\geq \tfrac{\ln x -1}{x}$ and
for $x\geq 3$, $\tfrac{\ln x -1}{x}>0$.

\begin{corollary}\label{cor:shearer}
Let $G$ be a triangle-free graph on $n$ vertices and maximum
degree~$\Delta\geq 3$. Then $G$ has an independent set of size at least
$\tfrac{\ln \Delta-1}{\Delta}n$.
\end{corollary}

We get the following lemma as a corollary, which we will use in the
proof of our main result.

\begin{lemma}\label{lemm:goodIS}
Let $G$ be an identifiable triangle-free graph on $n$ vertices and
maximum degree~$\Delta\geq 3$, and let $Y$ be the set of all vertices of $G$
having no false twin. Then $G[Y]$ has an independent set $S$ with the
following properties:
\begin{enumerate}
\item For each vertex $u$ of degree~1 in $G$, there exists a
  vertex of $G$ at distance~2 of $u$ which does not belong to $S$.\label{lemm:goodIS-prop-1}
\item $|S|\geq\tfrac{\ln\Delta-1}{\Delta}|Y|$\label{lemm:goodIS-prop-2}
\end{enumerate}
\end{lemma}
\begin{proof}
Let $S_1\subseteq Y$ be the set of vertices of $Y$ having degree~1 in
$G$. Note that since $G$ is identifiable, it has no isolated edges and
therefore $S_1$ is an independent set in $G$ (and $G[Y]$). Moreover
since $Y$ has no vertices having a false twin, all vertices of $S_1$
are at distance at least~3 from each other. Let $T_1$ be the set of
vertices constructed as follows. All the vertices of $S_1$ belong to
$T_1$. For each element $s$ of $S_1$, its unique neighbour in $G$
belongs to $T_1$, and some arbitrary neighbour at distance~2 of $s$
belongs to $T_1$. Since all the vertices of $S_1$ are at distance at
least~3 from each other, for each vertex $s$ of $S_1$ there is a
vertex at distance~2 of $s$ belonging to $T_1\setminus S_1$. We now set
$Y_1=T_1\cap Y$. Note that we have
$|S_1|\geq\tfrac{|T_1|}{3}\geq\tfrac{|Y_1|}{3}$ since for each vertex
of $S_1$, at most three vertices of $G$ have been inserted into $T_1$.

Now, let $Y_2=Y\setminus Y_1$. By the previous construction, $Y_2$
neither contains a vertex of degree~1 in $G$, nor a neighbour of
such a vertex. By Corollary~\ref{cor:shearer}, $G[Y_2]$ has an
independent set $S_2$ of size at least $\tfrac{\ln
  \Delta-1}{\Delta}|Y_2|$.

Taking $S=S_1\cup S_2$, we get an independent set of $G[Y]$ fulfilling
the first property of the claim. Moreover, $Y_1$ and $Y_2$ form a
partition of $Y$, $S_1\subseteq Y_1$ and $S_2\subseteq Y_2$. Since
for all strictly positive $x$, $\tfrac{1}{3}>\tfrac{\ln x -1}{x}$, we have:
$$|S|\geq\tfrac{|Y_1|}{3}+\tfrac{\ln
  \Delta-1}{\Delta}|Y_2|\geq\tfrac{\ln \Delta-1}{\Delta}|Y|$$
\end{proof}

\subsection{Quasi-identifying the vertices in and around a strong induced matching}\label{subsec:LR}

This subsection is devoted to the construction of small enough
quasi-identifying codes.

Recall that in order to prove our main result, given a nontrivial
identifiable connected triangle-free graph $G$, we will construct an
independent set $S$ and consider the (possibly empty) strong induced
matching $M$ such that $R(M)$ forms the set of isolated edges of
$V(G)\setminus S$. In order to ensure that there are no isolated edges
$uv$ in $G[V(G)\setminus S]$, it would suffice to remove an arbitrary
neighbour of either $u$ or $v$ from $S$. However, this could lead to a
very large identifying code. Indeed, consider the example of a
complete graph $K_n$ where each edge is subdivided twice, $K_n^*$. The
original vertices of $K_n$ form a (maximal) independent set $S$ and
each original edge of $K_n$ corresponds to an isolated edge in the
subgraph of $K_n^*$ induced by the complement of $S$,
$K_n^*[V(K_n^*)\setminus S]$. Now, in $K_n^*$, getting rid of all
isolated edges of $K_n^*[V(K_n^*)\setminus S]$ by removing vertices
from $S$ requires a vertex cover of $K_n$, that is, $n-1$
vertices. This would yield an identifying code of size $|V(K_n^*)|-1$,
which is not interesting.

Hence, in order to overcome this problem, we show in this subsection
how to build an $(L(M),R(M))$-quasi-identifying code of bounded size.
We first deal with the special case where all vertices of $R(M)$ have
degree exactly~2 (note that by Definition~\ref{def:LRmatch} they must
have degree at least~2).

\begin{lemma}\label{lemma:LRdeg2}
  Let $G$ be an identifiable (not necessarily connected) triangle-free
  graph having a strong induced matching $M$ where $L=L(M)$, $R=R(M)$,
  and all vertices of $R$ have degree exactly~2. Then, there is an
  $(L,R)$-quasi-identifying code $C$ of $G$ having the following
  properties:
\begin{enumerate}
\item $|C|\leq| L|+{\tfrac{|
      R|}{2}}$.\label{lemma:LRdeg2-1}
\item No vertex of $R$ is $C$-isolated.\label{lemma:LRdeg2-2}
\item At least half of the vertices of $L$ belong to $C$.\label{lemma:LRdeg2-3}
\end{enumerate}
\end{lemma}
\begin{proof}
  In order to simplify its construction, let us first define the
  multigraph $G_{L,R}=(L,E)$ with vertex set $L$ and in which there is
  an edge between two vertices $l_1,l_2$ of $L$ if and only if there
  exist two vertices $r_1,r_2$ of $R$, such that $l_1,r_1,r_2,l_2$ is
  a 3-path in $G$. In other words, we contract every path of length~3
  of $G[L\cup R]$ having both endpoints in $L$, into one edge. There
  can be multiple edges in $G_{L,R}$ (but no loops), since several
  disjoint 3-paths may join $l_1$ to $l_2$.

  From $G_{L,R}$ we will build an oriented multigraph
  $\overrightarrow{G}_{L,R}$. Given an orientation of
  $\overrightarrow{G}_{L,R}$, we define the subset
  $S(\overrightarrow{G}_{L,R})$ of vertices of $L\cup R$ in the
  following way: all the vertices of $L$ belong to
  $S(\overrightarrow{G}_{L,R})$, and for each arc
  $\overrightarrow{l_1l_2}$ of $\overrightarrow{G}_{L,R}$
  corresponding to the path $l_1r_1r_2l_2$ in $G$, $r_2$ belongs to
  $S(\overrightarrow{G}_{L,R})$. Note that
  $|S(\overrightarrow{G}_{L,R})|=|L|+\tfrac{|R|}{2}$. An illustration
  is given in Figure~\ref{fig:tfree_digraph}, where the gray vertices
  belong to $S(\overrightarrow{G}_{L,R})$. Our aim is to construct an
  orientation of $\overrightarrow{G}_{L,R}$ for which
  $S(\overrightarrow{G}_{L,R})$ is the desired $(L,R)$-quasi-identifying
  code of $G$.

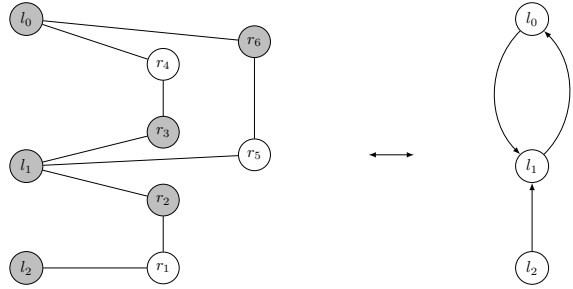
\begin{figure}[!ht]
\centering
\subfigure{
\scalebox{0.6}{\begin{tikzpicture}
\path (0,5.5) node[draw,shape=circle,fill=lightgray] (l0) {$l_0$};
\path (0,2.25) node[draw,shape=circle,fill=lightgray] (l1) {$l_1$};
\path (0,0) node[draw,shape=circle,fill=lightgray] (l2) {$l_2$};
\path (3,1.5) node[draw,shape=circle,fill=lightgray] (r1) {$r_2$};
\path (3,0) node[draw,shape=circle] (r2) {$r_1$};
\path (3,3) node[draw,shape=circle,fill=lightgray] (r3) {$r_3$};
\path (3,4.5) node[draw,shape=circle] (r4) {$r_4$};
\path (5,2.5) node[draw,shape=circle,] (r5) {$r_5$};
\path (5,5) node[draw,shape=circle,fill=lightgray] (r6) {$r_6$};

\draw (l1) -- (r1)-- (r2) -- (l2)
      (l1) -- (r3) -- (r4) -- (l0)
      (l1) -- (r5) -- (r6) -- (l0);
\draw[<->, >=latex, semithick] (7.5,2.5) -- +(1,0);
\end{tikzpicture}}
}\qquad
\subfigure{
\scalebox{0.6}{\begin{tikzpicture}
\path (0,5.5) node[draw,shape=circle] (l0) {$l_0$};
\path (0,2.25) node[draw,shape=circle] (l1) {$l_1$};
\path (0,0) node[draw,shape=circle] (l2) {$l_2$};
\draw[->,>=latex, semithick] (l0) .. controls +(-1,-1) and +(-1,1) .. (l1);
\draw[->,>=latex, semithick] (l1) .. controls +(1,1) and +(1,-1) .. (l0);
\draw[->,>=latex, semithick] (l2) -- (l1);
\end{tikzpicture}}}
\caption{Correspondance between a special subset of $L\cup R$ and $\protect\overrightarrow{G}_{L,R}$}
\label{fig:tfree_digraph}
\end{figure}

We start by orienting the arcs of $\overrightarrow{G}_{L,R}$ in an
arbitrary way. Note that $S(\overrightarrow{G}_{L,R})$ fulfills all
three required properties of the statement of the lemma.  Hence, if
$S(\overrightarrow{G}_{L,R})$ is an $(L, R)$-quasi-identifying code of
$G$, we are done. So suppose this is not the case. Note that
$S(\overrightarrow{G}_{L,R})$ fulfills conditions
number~\ref{def:quasi-1} and~\ref{def:quasi-3} of
Definition~\ref{def:quasi}. Hence, there are pairs of vertices of
$L\cup R$ which are not separated by
$S(\overrightarrow{G}_{L,R})$. The only case where a pair $l,r$ is not
separated by $S(\overrightarrow{G}_{L,R})$, is when $l\in L$, $r\in
R$, and both belong to $S(\overrightarrow{G}_{L,R})$, but they are
only dominated by each other and themselves. This is equivalent to the
case where $l$ is of in-degree~1 in $\overrightarrow{G}_{L,R}$ (see
Figure~\ref{fig:upperbnd_bnd_tfree-3} for an illustration). In this
case, in order to fix this problem, we modify the orientation of
$\overrightarrow{G}_{L,R}$ as follows.

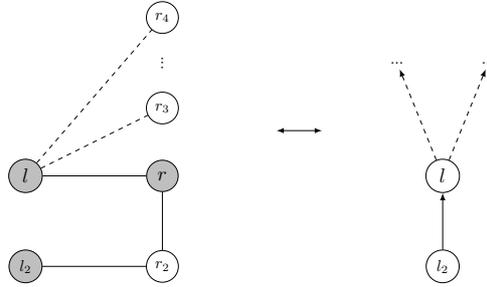
\begin{figure}[!ht]
\centering
\subfigure{\scalebox{0.6}{\begin{tikzpicture}
\path (0,1.5) node[draw,shape=circle,scale=1.25,fill=lightgray] (l1) {$l$};
\path (0,-0.5) node[draw,shape=circle,fill=lightgray] (l2) {$l_2$};
\path (3,1.5) node[draw,shape=circle,scale=1.25,fill=lightgray] (r1) {$r$};
\path (3,-0.5) node[draw,shape=circle] (r2) {$r_2$};
\path (3,3) node[draw,shape=circle] (r3) {$r_3$};
\path (3,5) node[draw,shape=circle] (r4) {$r_4$};
\draw (l1) -- (r1)
      (l2) -- (r2)
      (r1) -- (r2);

\draw[dashed] (l1) -- (r3)
        (l1) -- (r4);
\draw (3,4) node[rotate=90] {...};
\draw[<->, >=latex, semithick] (5.5,2.5) -- +(1,0);
\end{tikzpicture}}}\qquad
\subfigure{
\scalebox{0.6}{\begin{tikzpicture}
\path (1,4) node (l0) {...};
\path (-1,4) node (l00) {...};
\path (0,1.5) node[draw,scale=1.25,shape=circle] (l1) {$l$};
\path (0,-0.5) node[draw,shape=circle] (l2) {$l_2$};
\draw[->,>=latex, dashed, semithick] (l1) -- (l0);
\draw[->,>=latex, dashed, semithick] (l1) -- (l00);
\draw[->,>=latex, semithick] (l2) -- (l1);

\end{tikzpicture}}}
\caption{Vertices $l$ and $r$ are not separated}
\label{fig:upperbnd_bnd_tfree-3}
\end{figure}

At first, consider a connected component $\overrightarrow{G}_1$ of
$\overrightarrow{G}_{L,R}$, and construct an arbitrary spanning tree
$\overrightarrow{T}_1$ of $\overrightarrow{G}_1$, rooted in some
vertex $l$. Now, go through all vertices of $\overrightarrow{T}_1$,
level by level in a bottom-up order from the leaves up to the
root. Whenever the in-degree of the current vertex, $v$, is equal
to~1, swap the orientation of the arc joining $v$ to its parent in
$\overrightarrow{T}_1$. Doing so, the in-degree of $v$ in
$\overrightarrow{G}_1$ becomes distinct from~1, and the in-degree of
its parent is either incremented or decremented by~1. Note that except
for the root $l$, all vertices of $\overrightarrow{G}_1$ have now an
in-degree different from~1. This process is repeated for all connected
components of $\overrightarrow{G}_{L,R}$.

Let $C=S(\overrightarrow{G}_{L,R})$ be the new set corresponding to
the new orientation. If $C$ is an $(L,R)$-quasi-identifying code of
$G$, we are done. Otherwise, as observed earlier, it means that some
roots of the spanning trees we built, have in-degree~1 in
$\overrightarrow{G}_{L,R}$. Let $l$ be such a root with
in-degree~1. Observe that $l$ has a unique neighbour in $C\cap R$, say $r$. Let
$r_2$ be the neighbour of $r$ in $R$. It is
sufficient to take out $l$ from $C$ and to replace it by $r_2$ in
order to separate $l$ from $r$ in $G[L\cup R]$ (see
Figure~\ref{fig:upperbnd_bnd_tfree-5} for an illustration), without
changing the cardinality of $C$. Moreover, all neighbours of $l$ are
still separated from the other vertices because they are all in
$R\setminus C$ and therefore have a neighbour in $R\cap C$, which
itself has at least one neighbour in $L\cap C$. Hence $C$ is now an
$(L,R)$-quasi-identifying code of $G$. Since the process did not
change the cardinality of $C$, we get property
number~\ref{lemma:LRdeg2-1} of the claim of the lemma.

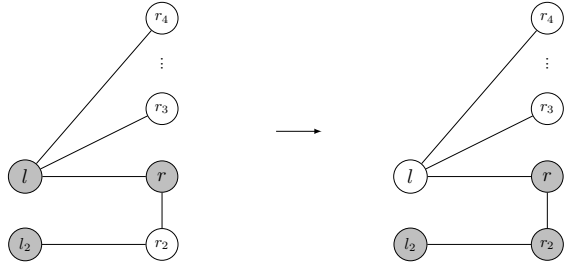
\begin{figure}[!ht]
\centering
\subfigure{
\scalebox{0.6}{\begin{tikzpicture}
\path (0,1.5) node[draw,shape=circle,scale=1.25,fill=lightgray] (l1) {$l$};
\path (0,0) node[draw,shape=circle,fill=lightgray] (l2) {$l_2$};
\path (3,1.5) node[draw,scale=1.25,shape=circle,fill=lightgray] (r1) {$r$};
\path (3,0) node[draw,shape=circle] (r2) {$r_2$};
\path (3,3) node[draw,shape=circle] (r3) {$r_3$};
\path (3,5) node[draw,shape=circle] (r4) {$r_4$};
\draw (l1) -- (r1)
      (l2) -- (r2)
      (r1) -- (r2)
      (l1) -- (r3)
      (l1) -- (r4);
\draw (3,4) node[rotate=90] {...};
\draw[->, >=latex, semithick] (5.5,2.5) -- +(1,0);
\end{tikzpicture}}
}\qquad
\subfigure{
\scalebox{0.6}{\begin{tikzpicture}
\path (0,1.5) node[draw,scale=1.25,shape=circle] (l1) {$l$};
\path (0,0) node[draw,shape=circle,fill=lightgray] (l2) {$l_2$};
\path (3,1.5) node[draw,scale=1.25,shape=circle,fill=lightgray] (r1) {$r$};
\path (3,0) node[draw,shape=circle,fill=lightgray] (r2) {$r_2$};
\path (3,3) node[draw,shape=circle] (r3) {$r_3$};
\path (3,5) node[draw,shape=circle] (r4) {$r_4$};
\draw (l1) -- (r1)
      (l2) -- (r2)
      (r1) -- (r2)
      (l1) -- (r3)
      (l1) -- (r4);
\draw (3,4) node[rotate=90] {...};
\end{tikzpicture}}}
\caption{Local modification of the constructed code}
\label{fig:upperbnd_bnd_tfree-5}
\end{figure}

Notice that there are at most $\tfrac{|L|}{2}$ connected
components in $G[L\cup R]$ since each of them contains at least two
vertices of $L$. Thus property number~\ref{lemma:LRdeg2-3} of the
claim of the lemma follows.

Property number~\ref{lemma:LRdeg2-2} is fulfilled by the construction
of $C$ since in each pair of adjacent vertices of $R$, either it has a
code vertex in $L$ as a neighbour if there was no modification done,
or in $R$ if a switch of two elements of $L$ and $R$ was
necessary. Moreover, for each such pair, at least one of its elements
belongs to the code. This shows that $C$ is an
$(L,R)$-quasi-identifying code and completes the proof.
\end{proof}

We now deal with the general case, where the vertices of
$R(M)$ have degree \emph{at least}~2 as required in
Definition~\ref{def:LRmatch}.

\begin{lemma}\label{lemma:LR}
  Let $G$ be an identifiable (not necessarily connected) triangle-free
  graph having a strong induced matching $M$, with $L=L(M)$ and
  $R=R(M)$. There exists a set $L'$ of vertices of $L\cup R$ such that
  $|L'| \ge \frac{|L|}{3}$, and $C=(L\cup R)\setminus L'$ is an
  $(L,R)$-quasi-identifying code of $G$ having no $C$-isolated
  vertices.
\end{lemma}
\begin{proof}
  Let us first divide sets $L$ and $R$ into the following subsets:
  let $R_1\subseteq R$ be such that $r\in R_1$ if both $r$ and its
  unique neighbour in $R$ are of degree~2. Let $L_1\subseteq L$ be the
  set of all neighbours of vertices of $R_1$, let $R_2=R\setminus
  R_1$, and let $L_2=L\setminus L_1$ (see
  Figure~\ref{fig:upperbnd_bnd_tfree-2} for an illustration).

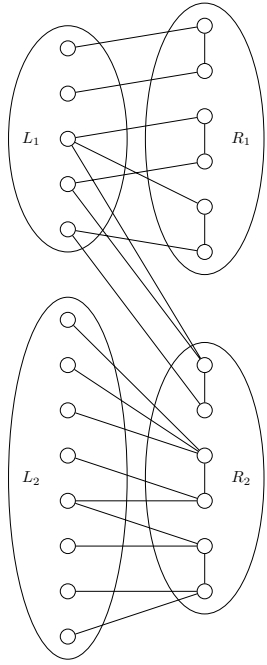
\begin{figure}
\centering
\scalebox{0.6}{\begin{tikzpicture}
\path (0,0) node[draw,shape=circle] (l21) {};
\path (0,1) node[draw,shape=circle] (l22) {};
\path (0,2) node[draw,shape=circle] (l23) {};
\path (0,3) node[draw,shape=circle] (l24) {};
\path (0,4) node[draw,shape=circle] (l25) {};
\path (0,5) node[draw,shape=circle] (l26) {};
\path (0,6) node[draw,shape=circle] (l27) {};
\path (0,7) node[draw,shape=circle] (l28) {};
\path (0,9) node[draw,shape=circle] (l11) {};
\path (0,10) node[draw,shape=circle] (l12) {};
\path (0,11) node[draw,shape=circle] (l13) {};
\path (0,12) node[draw,shape=circle] (l14) {};
\path (0,13) node[draw,shape=circle] (l15) {};
\path (3,1) node[draw,shape=circle] (r21) {};
\path (3,2) node[draw,shape=circle] (r22) {};
\path (3,3) node[draw,shape=circle] (r23) {};
\path (3,4) node[draw,shape=circle] (r24) {};
\path (3,5) node[draw,shape=circle] (r25) {};
\path (3,6) node[draw,shape=circle] (r26) {};
\path (3,8.5) node[draw,shape=circle] (r11) {};
\path (3,9.5) node[draw,shape=circle] (r12) {};
\path (3,10.5) node[draw,shape=circle] (r13) {};
\path (3,11.5) node[draw,shape=circle] (r14) {};
\path (3,12.5) node[draw,shape=circle] (r15) {};
\path (3,13.5) node[draw,shape=circle] (r16) {};
\draw (l21) -- (r21)
      (l22) -- (r21)
      (l23) -- (r22)
      (l24) -- (r22)
      (l24) -- (r23)
      (l25) -- (r23)
      (l26) -- (r24)
      (l27) -- (r24)
      (l28) -- (r24)
      (l11) -- (r25)
      (l12) -- (r26)
      (l13) -- (r26)
      (l11) -- (r11)
      (l12) -- (r13)
      (l13) -- (r12)
      (l13) -- (r14)
      (l14) -- (r15)
      (l15) -- (r16)
      (r21) -- (r22)
      (r23) -- (r24)
      (r25) -- (r26)
      (r11) -- (r12)
      (r13) -- (r14)
      (r15) -- (r16);
\draw (0,11) ellipse (1.3cm and 2.5cm)
      (0,3.5) ellipse (1.3cm and 4cm)
      (3,11) ellipse (1.3cm and 3cm)
      (3,3.5) ellipse (1.3cm and 3cm)
      (-0.8,11) node {$L_1$}
      (-0.8,3.5) node {$L_2$}
      (3.8,11) node {$R_1$}
      (3.8,3.5) node {$R_2$};
\end{tikzpicture}}
\caption{Illustration of sets $L_1$, $L_2$, $R_1$, and $R_2$}
\label{fig:upperbnd_bnd_tfree-2}
\end{figure}

We can use Lemma~\ref{lemma:LRdeg2} to construct an
$(L_1,R_1)$-quasi-identifying code $C_1$ of $G$ such that the three
properties described in the statement of Lemma~\ref{lemma:LRdeg2} are fulfilled. Let
$C_1$ be such a code, in particular we have $|C_1|\le |L_1|
+\tfrac{|R_1|}{2}$. Let us now describe the construction of two
distinct $(L,R)$-quasi-identifying codes $C_a$ and $C_b$.

\begin{itemize}
\item\textbf{Construction of code $C_a$.}\\We construct $C_a$
  such that
  $|C_a|\leq|L_1|+\tfrac{|R_1|}{2}+|L_2|+\tfrac{|R_2|}{2}+\min\left\{\tfrac{|L_1|}{2},\tfrac{|R_2|}{2}\right\}$,
  as follows.
 \begin{enumerate}
 \item Put $C_1$ into $C_a$.
 \item Put $L_2$ into $C_a$.
 \item For each pair $r,r'$ of adjacent vertices of $R_2$, let $r^*$
   be one of them having at least two neighbours in $L$ (by definition
   of $R_2$ either $r$ or $r'$ has this property). Put $r^*$ into
   $C_a$.\label{constr-ca-3}
 \item For each pair $r,r'$ of adjacent vertices of $R_2$, let $r^*$ be
   the one which was put into $C_a$ in the previous step. Check if $r^*$
   has less than two neighbours within $C_a\cap L$ (this may happen if
   some of its neighbours are in $L_1$, and they do not belong to
   $C_1$). If this is the case, pick an additional neighbour of $r^*$ --- which
   exists since $r$ has at least two neighbours in $L$ --- and put it
   into $C_a$. Note that this is done at most $\tfrac{|R_2|}{2}$
   times. Moreover, at most $\tfrac{|L_1|}{2}$ new vertices from $L_1$
   are put into $C_a$ in such a way since by property
   number~\ref{lemma:LRdeg2-3} of Lemma~\ref{lemma:LRdeg2}, there are
   at most $\tfrac{|L_1|}{2}$ vertices of $L_1$ not in
   $C_1$.\label{constr-ca-4}
 \item Finally, consider each $C_a$-isolated vertex $l$ of $L$, take it out
   of $C_a$ and put an arbitrary neighbour of $l$ into $C_a$ (this
   operation does not affect the size of $C_a$).
\end{enumerate}
\item\textbf{Construction of code $C_b$.}\\We construct $C_b$
  such that $|C_b|\leq |L_1|+\tfrac{|R_1|}{2}+3\tfrac{|R_2|}{2}$, as
  follows.
 \begin{enumerate}
 \item Put $C_1$ into $C_b$.
 \item Put $R_2$ into $C_b$\label{constr-cb-2}.
 \item For each pair $r,r'$ of adjacent vertices of $R_2$, one
   arbitrary neighbour in $L$ of either $r$ or $r'$ is put into $C_b$.\label{constr-cb-3}
 \item Finally, in the same way as for the construction of $C_a$, we
   get rid of each $C_b$-isolated vertex $l$ of $L$ by taking $l$ out
   of $C_b$ and putting an arbitrary neighbour of $l$ into $C_b$
   instead.
\end{enumerate}
\end{itemize}

Let us now prove that $C_a$ and $C_b$ are $(L,R)$-quasi-identifying
codes without $C_a$-isolated or $C_b$-isolated vertices. First note that in
both constructions, the final step consists in replacing some
$C_a$-isolated vertices from $C_a$ (resp.  $C_b$). In order to
simplify the proof, let $C_a^*$ (resp.  $C_b^*$) be the code as it is
before this last step. We first prove that $C_a^*$ (resp. $C_b^*$)
have all desired properties except that there remain $C_a^*$-isolated
(resp. $C_b^*$-isolated) vertices in $L$. We then prove that
performing the last step transforms it into an
$(L,R)$-quasi-identifying code with all required properties.

It can first be noticed that both $C_a^*$ and $C_b^*$ are dominating
sets, so point number~\ref{def:quasi-1} of Definition~\ref{def:quasi}
holds.

Let us now show point number~\ref{def:quasi-2} of
Definition~\ref{def:quasi} (the separation condition). In both codes,
the vertices of all pairs $u,v$ of vertices of $L_1\cup R_1$ are
separated from each other, since $C_1$ is a subset of both $C_a^*$ and
$C_b^*$.

Now, suppose that $u\in R_1$ and $v\in L_2\cup R_2$. By definition
of $R_1$, no vertex of $R_1$ is adjacent to any vertex of $L_2\cup
R_2$. Therefore, by condition number~\ref{def:quasi-3} of
Definition~\ref{def:quasi}, either $u$ or its neighbour in $R_1$
belong to $C_1$, hence $u$ and $v$ are separated.

Thus, it remains to check if $u$ and $v$ are separated when $u\in L_1$
and $v\in L_2\cup R_2$, and when both $u$ and $v$ belong to $L_2\cup
R_2$. We deal with $C_a^*$ and $C_b^*$ separately.

\paragraph{Code $C_a^*$}
\begin{itemize}
\item Suppose $u\in L_1$ and $v\in L_2\cup R_2$. Note that $u$ is
  dominated by some vertex $x$ within $L_1\cup R_1$ since
  $C_1\subseteq C_a^*$. If $v\in L_2$, $u$ and $v$ are separated by $x$
  since no vertex of $L_2$ is adjacent to any vertex of $L_1\cup
  R_1$. If $v\in R_2$ and $v\notin C_a^*$, then $u$ and $v$ are
  separated by the neighbour of $v$ in $R_2$, which belongs to
  $C_a^*$. Similarly, if $u$ has a neighbour in $R_1$ belonging to
  $C_1$, we are done. Otherwise, it means that $v\in C_a^*$ and $u\in
  C_1$ (otherwise $u$ would not be dominated by $C_1$). Hence $v$ has
  another neighbour in $L$, say $u'$, belonging to $C_a^*$, and $u'$
  separates $u$ from $v$. Indeed, at step~\ref{constr-ca-4} of the
  construction of $C_a$, either $v$ already had at least two
  neighbours in $L\cap C_a^*$, or an additional one has been added.

\item Now, suppose both $u$ and $v$ belong to $L_2\cup R_2$.

If both $u$ and $v \in L_2$, they are separated since the whole set
$L_2$, which is independent, belongs to $C_a^*$.

If both $u$ and $v$ belong to $R_2$ and they are not adjacent, they
are separated since either themselves or their respective neighbours
in $R_2$ belong to $C_a^*$ by step~\ref{constr-ca-3} of its
construction. Otherwise, for the same reason one of them (say $u$)
belongs to the code. It is ensured in step~\ref{constr-ca-4} that at
least one neighbour of $u$ in $L$ belongs to $C_a^*$, therefore $u$
and $v$ are separated by this neighbour.

If $u\in L_2$ and $v\in R_2$ and they are not adjacent, they are
separated by $u$ since the whole set $L_2$ belongs to
$C_a^*$. Otherwise, if $v\notin C_a^*$, they are separated by the
neighbour of $v$ in $R_2$. Otherwise, again by step~\ref{constr-ca-4}
of the construction $v$ has a second neighbour in $L\cap C_a^*$,
separating them.
\end{itemize}

\paragraph{Code $C_b^*$}
\begin{itemize}
\item If $u\in L_1$ and $v\in L_2\cup R_2$, $u$ and $v$ are separated by a
neighbour of $v$ belonging to $R_2$ since the whole set $R_2$ is in
$C_b^*$.

\item Now, suppose $u,v\in L_2\cup R_2$.

  If both $u,v$ belong to $L_2$, and they have the same set of
  neighbours within $R$, we are done since they do not need to be
  separated (point number~\ref{def:quasi-2} of
  Definition~\ref{def:quasi}). Otherwise, they are separated since all
  their neighbours within $L\cup R$ belong to $R_2$, and
  $R_2\subseteq C_b^*$.

  If both $u,v$ belong to $R_2$, $u$ and $v$ are separated by
  themselves if they are not adjacent. Otherwise, they are separated
  by a neighbour of one of them in $L\cap C_b^*$, added at
  step~\ref{constr-cb-3} of the construction.

Finally, if $u\in R_2$ and $v\in L_2$, then $u$ and $v$ are either
separated by $u$ if $u$ and $v$ are not adjacent, or by the neighbour
of $u$ in $R_2$ otherwise.
\end{itemize}

Let us now check point number~\ref{def:quasi-3} of
Definition~\ref{def:quasi}, i.e. that for each pair of adjacent
vertices in $R$, at least one of them belongs to the code. This is
true for vertices of $R_1$ since $C_1$ is an
$(L_1,R_1)$-quasi-identifying code and therefore fulfills this
condition. This is also ensured for vertices of $R_2$ at
step~\ref{constr-ca-3} of the construction of $C_a$ and
at step~\ref{constr-cb-2} of the construction of $C_b$.

Hence, we have shown that both $C_a^*$ and $C_b^*$ are
$(L,R)$-quasi-identifying codes.

Moreover, there are no $C_a^*$-isolated (resp. $C_b^*$-isolated)
vertices in $R$: there are no such vertices in $R_1$ by
Lemma~\ref{lemma:LRdeg2}, and no such vertices in $R_2$ for $C_a^*$ by
step~\ref{constr-ca-4} of its construction, and for $C_b^*$ as well
since $R_2\subseteq C_b^*$.

\vspace{0.2cm}
As announced previously, we now have to deal with the last step of the
constructions of both $C_a$ and $C_b$. It is easily observed that this
step does not affect the domination property of both codes. Indeed,
the former $C_a$-,$C_b$-isolated vertices themselves are now dominated
by some neighbour. Moreover each of their neighbours belongs to $R$,
and since $C_a$ and $C_b$ are $(L,R)$-quasi-identifying its own
neighbour in $R$ belongs to the code.

Let us prove that the separation condition is still satisfied by $C_a$
and $C_b$. Let $C_x$ ($x\in\{a,b\}$) be the considered code and let
$l\in L$ be a $C_x$-isolated vertex which gets replaced in $C_x$ by
one of its neighbours in $R$, say $r_l$. The only vertices which might
be affected by the modification, are vertices which were previously
dominated by $l$, i.e. vertices of $B(l)$: assume, by contradiction,
that $u\in B(l)$ is no longer separated from some vertex $v$.

If $u=l$, in $C_x$, we have $B(l)\cap C_x=\{r_l\}$. Since $B(v)\cap
C_x=\{r_l\}$ and the neighbour of $r_l$ in $R$ belongs to $C_x$, $v\in
L$. Moreover, observe that $v$ was dominated by a vertex of
$C_x^*$, say $v'$, and $v'\notin B(l)$ since $l$ is
$C_x^*$-isolated. Hence, it means that $v$ was also
$C_x^*$-isolated. But then, in the last step of the construction of
$C_x$, one of $l$ and $v$, say $l$, has been considered first and
replaced by $r_l$, leaving them separated by $v'$, a contradiction.

Now, if $u$ is a neighbour of $l$, $u\in R$ and the neighbour of $u$
in $R$, call him $u'$, belongs to $C_x$ by construction. Since $C_x^*$
is an $(L,R)$-quasi-identifying code, $u'$ has a neighbour belonging
to $L$ and to the code. Hence $u$ and $u'$ are separated, $u\neq r_l$
and $v$ must be a neighbour of $u'$ not belonging to the code. Hence
$u\in R_2$ since $u'$ has degree at least~3. Moreover, $v\in L_2$;
otherwise, since $C_1\subseteq C_x$, $v$ would be dominated within
$C_1$ and $u,v$ would be separated --- a contradiction. Now, if
$C_x=C_a$, $v\in C_a$, a contradiction. If $C_x=C_b$, $u\in C_b$, a
contradiction too. This completes the proof of the separation
property.

Now, note that point number~\ref{def:quasi-3} of
Definition~\ref{def:quasi} remains verified as no vertex of $R$ is
removed from neither $C_a$ or $C_b$ in the last step of their
construction. Finally, observe that thanks to the last step of the
constructions, there are no $C_x$-isolated ($x\in\{a,b\}$) vertices in
$L$ anymore. Moreover, this step has not created any $C_x$-isolated
vertices in $R$. Indeed, the vertices which are added, did not belong
to $C_x^*$, and hence their neighbour in $R$ did. This completes the
proof of the validity of both constructions $C_a$ and $C_b$.

\vspace{0.2cm}
Let us now determine a lower bound on the cardinality of $(L\cup
R)\setminus C_x$, for $x\in\{a,b\}$. Taking into account that
$|L_1|\le|R_1|$, we obtain:
\begin{align*}
    |(L\cup R) \setminus C_a| & \geq |L_1|+|L_2|+|R_1|+|R_2|-|C_a|\\
    & \geq \frac{|R_1|}{2} + \frac{|R_2|}{2} -\min\left\{\frac{|L_1|}{2},\frac{|R_2|}{2}\right\}
\end{align*}

Thus, both following equations hold:
\begin{align}
  |(L\cup R)\setminus C_a| &\geq \frac{|R_1|}{2}+\frac{|R_2|}{2}-\frac{|L_1|}{2} \geq \frac{|R_2|}{2}  \label{eqn:Ca-1}\\
  |(L\cup R)\setminus C_a| & \geq \frac{|R_1|}{2}+
  \frac{|R_2|}{2}-\frac{|R_2|}{2} = \frac{|R_1|}{2}\geq
  \frac{|L_1|}{2}\label{eqn:Ca-2}
\end{align}

Similarly,

\begin{align}
\begin{split}
    |(L\cup R) \setminus C_b| & \geq |L_1|+|L_2|+|R_1|+|R_2|-|C_b|\\
    & \geq |L_2|+\frac{|R_1|}{2}-\frac{|R_2|}{2}\\
    & \geq |L_2|+\frac{|L_1|}{2}-\frac{|R_2|}{2}\\
    & = |L|-\frac{|L_1|}{2}-\frac{|R_2|}{2}\label{eqn:Cb}
\end{split}
\end{align}

Hence intuitively, the previous equations show that our two codes fit
to two different situations: $C_a$ is useful when either $|L_1|$ or
$|R_2|$ is large enough compared to $|L|$, whereas $C_b$ is useful
when $|L_1|+|R_2|$ is small enough compared to $|L|$. Let $C
\in\{C_a,C_b\}$ be the code having the minimum cardinality. Then,
using inequalities~\eqref{eqn:Ca-1}, \eqref{eqn:Ca-2}
and~\eqref{eqn:Cb} and denoting $b = {\tfrac{\max\big\{|L_1|,
    |R_2|\big\}}{|L|}}$ we get:

\begin{align*}
  |(L\cup R) \setminus C| & \geq \max\left\{\frac{|L_1|}{2}, \frac{|R_2|}{2}, |L|-\frac{|L_1|}{2}-\frac{|R_2|}{2} \right\}\\
  & = \frac{|L|}{2}\cdot\max\left\{\frac{|L_1|}{|L|}, \frac{|R_2|}{|L|}, 2-\frac{|L_1|+|R_2|}{|L|}\right\}\\
& \geq \frac{|L|}{2}\cdot\max\left\{\frac{\max\left\{|L_1|, |R_2|\right\}}{|L|}, 2-\frac{2\cdot\max\left\{|L_1|,|R_2|\right\}}{|L|}\right\}\\
& = \frac{|L|}{2}\cdot \max\left\{b, 2-2b\right\}\\
& \geq \frac{|L|}{2}\cdot\min\limits_{b\geq 0}\left\{\max\left\{b,
2-2b\right\}\right\}
\end{align*}

Note that $\min\limits_{b\geq 0}\left\{\max\left\{b,
    2-2b\right\}\right\}=\tfrac{2}{3}$. Hence, we get:

$$
|(L\cup R) \setminus
C|\geq\frac{|L|}{2}\cdot\frac{2}{3}=\frac{|L|}{3}
$$

Note that equality in the previous inequality is achieved when
$|L_1|=|R_1|=|R_2|=2|L_2|$.

Putting $L' = (L\cup R) \setminus C$, we obtain the claim of the
lemma.
\end{proof}

\subsection{The main result}\label{subsec:mainresult}

We are now ready to prove the main theorem of this paper. The proof
has been sketched in Algorithm~\ref{alg:global_alg}, we now provide
all the details.

\begin{theorem}\label{th:upperbound_notriangle-new}
  Let $G$ be a connected identifiable triangle-free graph
  on $n$ vertices with maximum degree~$\Delta\geq 3$. Then $\M(G)\leq
  n-\frac{n}{\Delta+\tfrac{3\Delta}{\ln\Delta-1}}=n-\frac{n}{\Delta+o(\Delta)}$.
\end{theorem}
\begin{proof}
Let $\mathcal{F}=\{F_1,\ldots,F_{|\mathcal{F}|}\}$ be the set of all
nontrivial equivalence classes over the false twin relation $\equiv$
over $V(G)$. Let $X=\cup_{i=1}^{|\mathcal{F}|}F_i$ and $Y=V(G)\setminus
X$. We distinguish two cases.
\vspace{0.5cm}

\textbf{Case 1:} $|Y|\geq\tfrac{3n}{\ln\Delta+2}$.\\ In this case, let
$S$ be an independent set of $G[Y]$ given by Lemma~\ref{lemm:goodIS}:
we have $|S|\geq
\tfrac{\ln\Delta-1}{\Delta}|Y|\geq\tfrac{3n(\ln\Delta-1)}{\Delta(\ln\Delta+2)}$. Consider
all pairs $u,v$ of vertices of $G$ such that $u$ and $v$ are adjacent,
both $u$ and $v$ have degree at least~2, and all the vertices of
$N(u)\cup N(v)\setminus\{u,v\}$ belong to $S$ (see
Figure~\ref{fig:upperbnd_bnd_tfree-1} for an illustration). Since all
neighbours of $u$ and $v$ (except $u$ and $v$ themselves) are in $S$,
these neighbours form an independent set. Let $M$ be the (possibly
empty) set of all edges $uv$ such that $u$ and $v$ form such a
pair. By the previous remark, $M$ is a strong induced matching of
$G$. Let us denote $L=L(M)$ and $R=R(M)$. Note that we have
$L(M)\subseteq S$.

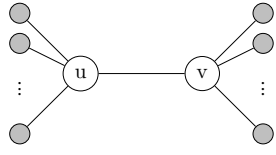
\begin{figure}
\centering
\scalebox{0.8}{\begin{tikzpicture}
\path (-1,1) node[draw,shape=circle,fill=lightgray] (a1) {};
\path (-1,0.5) node[draw,shape=circle,fill=lightgray] (a2) {};
\path (-1,-1) node[draw,shape=circle,fill=lightgray] (a3) {};
\path (3,1) node[draw,shape=circle,fill=lightgray] (b1) {};
\path (3,0.5) node[draw,shape=circle,fill=lightgray] (b2) {};
\path (3,-1) node[draw,shape=circle,fill=lightgray] (b3) {};
\path (0,0) node[draw,shape=circle] (u) {u};
\path (2,0) node[draw,shape=circle] (v) {v};
\draw (a1) -- (u)
      (a2) -- (u)
      (a3) -- (u)
      (v) -- (b1)
      (v) -- (b2)
      (v) -- (b3)
      (u) -- (v);
\draw (-1,-0.25) node[rotate=90] {...}
       (3,-0.25) node[rotate=90] {...};
\end{tikzpicture}}
\caption{Vertices $u,v$ with $(N(u)\cup N(v))\setminus\{u,v\}\subseteq S$}
\label{fig:upperbnd_bnd_tfree-1}
\end{figure}

Let us now partition $V(G)$ into two subsets of vertices: $L\cup R$ on the one hand, and
$V(G)\setminus(L\cup R)$ on the other hand. Such a partition is
illustrated in Figure~\ref{fig:partitionLR}. Note that $G[L\cup R]$ is
identifiable by Observation~\ref{obs:LRM-twinfree}. Let us show that
$G[V(G)\setminus(L\cup R)]$ is also identifiable. By contradiction,
suppose it is not the case and let $u,v$ be a pair of vertices such
that $B_{G[V(G)\setminus(L\cup R)]}(u)=B_{G[V(G)\setminus(L\cup
    R)]}(v)$. Vertices $u$ and $v$ are therefore adjacent, and since
$G$ is triangle-free, neither $u$ nor $v$ has other neighbours within
$G[V(G)\setminus(L\cup R)]$. Since $G$ is identifiable, at least one
of them has a neighbour in $L$. Suppose they both have a neighbour in
$L$. Then by construction of $S$, $u$ and $v$ both do not belong to
$S$. But then $u$ and $v$ should belong to $R$, a contradiction. Thus,
one of them, say $u$, has degree~1 in $G$, and all neighbours of $v$
belong to $L\subseteq S$. But by the first property of
$S$ in Lemma~\ref{lemm:goodIS}, at least one vertex at distance~2 of $u$ does
not belong to $S$, a contradiction.

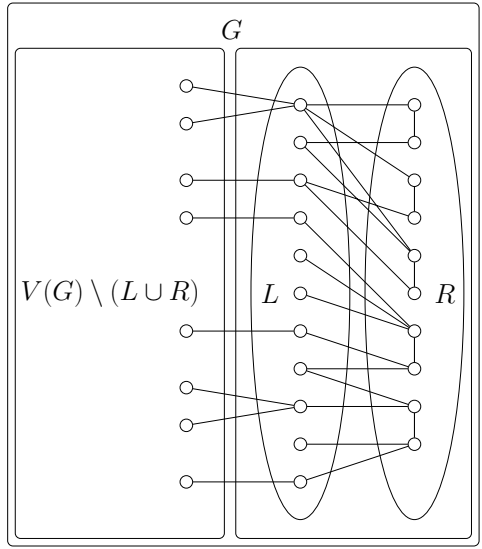
\begin{figure}
\centering
\scalebox{0.5}{\begin{tikzpicture}
\path (0,0) node[draw,shape=circle] (l21) {};
\path (0,1) node[draw,shape=circle] (l22) {};
\path (0,2) node[draw,shape=circle] (l23) {};
\path (0,3) node[draw,shape=circle] (l24) {};
\path (0,4) node[draw,shape=circle] (l25) {};
\path (0,5) node[draw,shape=circle] (l26) {};
\path (0,6) node[draw,shape=circle] (l27) {};
\path (0,7) node[draw,shape=circle] (l28) {};
\path (0,8) node[draw,shape=circle] (l11) {};
\path (0,9) node[draw,shape=circle] (l12) {};
\path (0,10) node[draw,shape=circle] (l13) {};

\path (3,1) node[draw,shape=circle] (r21) {};
\path (3,2) node[draw,shape=circle] (r22) {};
\path (3,3) node[draw,shape=circle] (r23) {};
\path (3,4) node[draw,shape=circle] (r24) {};
\path (3,5) node[draw,shape=circle] (r25) {};
\path (3,6) node[draw,shape=circle] (r26) {};
\path (3,7) node[draw,shape=circle] (r11) {};
\path (3,8) node[draw,shape=circle] (r12) {};
\path (3,9) node[draw,shape=circle] (r13) {};
\path (3,10) node[draw,shape=circle] (r14) {};

\path (-3,0) node[draw,shape=circle] (m0) {};
\path (-3,1.5) node[draw,shape=circle] (m1) {};
\path (-3,2.5) node[draw,shape=circle] (m2) {};
\path (-3,4) node[draw,shape=circle] (m3) {};
\path (-3,7) node[draw,shape=circle] (m4) {};
\path (-3,8) node[draw,shape=circle] (m5) {};
\path (-3,9.5) node[draw,shape=circle] (m6) {};
\path (-3,10.5) node[draw,shape=circle] (m7) {};

\draw (l21) -- (r21)
      (l22) -- (r21)
      (l23) -- (r22)
      (l24) -- (r22)
      (l24) -- (r23)
      (l25) -- (r23)
      (l26) -- (r24)
      (l27) -- (r24)
      (l28) -- (r24)
      (l11) -- (r25)
      (l12) -- (r26)
      (l13) -- (r26)
      (l11) -- (r11)
      (l12) -- (r13)
      (l13) -- (r12)
      (l13) -- (r14);
\draw (l13) -- (m7);
\draw (l13) -- (m6);
\draw (l11) -- (m5);
\draw (l28) -- (m4);
\draw (l25) -- (m3);
\draw (l23) -- (m2);
\draw (l23) -- (m1);
\draw (l21) -- (m0);

\draw[thick]
      (r21) -- (r22)
      (r23) -- (r24)
      (r25) -- (r26)
      (r11) -- (r12)
      (r13) -- (r14);
\draw (0,5.0) ellipse (1.3cm and 6.0cm)
      (3,5.0) ellipse (1.3cm and 6.0cm)
      (-0.8,5) node {{\huge$L$}}
      (3.8,5) node {{\huge$R$}}
      (-5,5) node {{\huge$V(G)\setminus(L\cup R)$}}
      (-1.8,12) node {{\huge$G$}};
\draw[rounded corners] (-7.7,-1.7) rectangle (4.7,12.7);
\draw[rounded corners] (-7.5,-1.5) rectangle (-2,11.5);
\draw[rounded corners] (-1.7,-1.5) rectangle (4.5,11.5);
\end{tikzpicture}}
\caption{Partition of $V(G)$}
\label{fig:partitionLR}
\end{figure}

We will now build two subsets $C_1\subseteq L\cup R$ and
$C_2\subseteq V(G)\setminus(L\cup R)$ such that $C=C_1\cup C_2$ is an
identifying code of $G$.

\begin{itemize}
\item \textbf{Building $C_1\subseteq L\cup R$.}\\
  If $L\cup R=\emptyset$ we take $C_1=\emptyset$. Otherwise, we build
  $C_1$ using Lemma~\ref{lemma:LR}: applying it to $G$ and $M$, we
  know that there exists an $(L,R)$-quasi-identifying code $C_1$ of
  $G$ without $C_1$-isolated vertices. From Lemma~\ref{lemma:LR} we
  also know that $|L'|\geq\tfrac{|L|}{3}$, where $L'=(L\cup
  R)\setminus C_1$.

\item \textbf{Building $C_2\subseteq V(G)\setminus(L\cup R)$.}\\
  Again if $V(G)\setminus(L\cup R)=\emptyset$ we take $C_2=\emptyset$. Otherwise, we take $C_2$ to be the complement of $S$ in
$V(G)\setminus(L\cup R)$: $C_2=\left(V(G)\setminus (L\cup
R)\right)\setminus S$. Let us show that $C_2$ is a
$\big(V(G)\setminus(L\cup R)\big)$-identifying code of $G$.

First, recall that $G'=G[V(G)\setminus(L\cup R)]$ is
identifiable. Note that $S$ does not contain any vertex $v$ which is
isolated in $G'$. Indeed, $G$ does not contain any isolated vertex,
hence if $v$ is isolated in $G'$, $v$ has a neighbour in $L$. But
$L\subseteq S$, a contradiction since $S$ is an independent set. We
also claim that for each vertex $v$ of degree~1 in $G'$, there is a
vertex at distance~2 of $v$ in $G'$ not belonging to $S$. Let $w$ be
the unique neighbour of $v$ in $G'$. If $v$ is also of degree~1 in
$G$, since $G'$ has no pair of twins, by the first property of $S$ in
Lemma~\ref{lemm:goodIS}, $w$ must have a neighbour $x$ not in
$S$. Vertex $x$ cannot belong to $L$, hence it belongs to $G'$ and we
are done. Now, if $v$ is not of degree~1 in $G$, all its neighbours in
$G$ other than $w$ belong to $L$. But since $G'$ is identifiable, $w$
has at least one neighbour other than $v$, belonging to $G'$ but not
to $S$, since otherwise $v$ and $w$ would belong to set $R$. Finally,
by construction of $G'$, there are no isolated edges in
$G[V(G')\setminus S]$.

Under these conditions we can apply Proposition~\ref{prop:IScode} on
$G'$ and on set $S$ restricted to $V(G')$, which shows that $C_2$ is a
$\big(V(G)\setminus(L\cup R)\big)$-identifying code of $G$.
\end{itemize}

We now have an $(L,R)$-quasi-identifying code $C_1$ of $G$ without
$C_1$-isolated vertices, and showed that $C_2$ is a
$(V(G)\setminus(L\cup R))$-identifying code of $G$. Moreover, $S$ does
not contain any pair of false twins. Furthermore, since $C_2$ is the
complement of $S$ in $G[V(G)\setminus(L\cup R)]$, all neighbours of
$L$ in $G[V(G)\setminus(L\cup R)]$ belong to $C_2$. Therefore, we can
apply Proposition~\ref{prop:well-id} and $C=C_1\cup C_2$ is an
identifying code of $G$.

Let us now upper-bound the size of $C$. To this end, we lower-bound
the size of its complement. From the construction of $C_1$ and $C_2$,
we have $V(G)\setminus C=(S\setminus L)\cup L'$.

Since $L\subseteq S$ and $|L'|\geq\tfrac{|L|}{3}$, we have
$|(S\setminus L)\cup L'|\geq\tfrac{|S|}{3}$.

Hence, we get: 
\begin{eqnarray*}
|V(G)\setminus C| & \geq & \tfrac{|S|}{3}\\
& \geq & \tfrac{\ln\Delta-1}{\Delta(\ln\Delta+2)}n\\
& = & \tfrac{n}{\Delta\tfrac{\ln\Delta+2}{\ln\Delta-1}}\\
& = & \tfrac{n}{\Delta+\tfrac{3\Delta}{\ln\Delta -1}}
\end{eqnarray*}

Hence, $|C|\leq n-\tfrac{n}{\Delta+\tfrac{3\Delta}{\ln\Delta -1}}$.

\vspace{0.5cm} \textbf{Case 2:} $|Y|\leq\tfrac{3n}{\ln\Delta+2}$.\\ Then,
$|X|\geq n-\tfrac{3n}{\ln\Delta+2}$. Since each set of $\mathcal{F}$ has size at most $\Delta$, we have:
\begin{eqnarray*}
 |\mathcal{F}| & \geq & \tfrac{|X|}{\Delta}\\
 & \geq & \tfrac{\ln\Delta-1}{\Delta(\ln\Delta+2)}n\\
 & = & \tfrac{n}{\Delta+\tfrac{3\Delta}{\ln\Delta-1}}
\end{eqnarray*}
Since $\Delta\geq 3$, $G$ is not isomorphic to $C_4$ and we can apply
Proposition~\ref{lemm:falsetwincode}: $G$ has an identifying code of
size at most $n-|\mathcal{F}|\leq
n-\tfrac{n}{\Delta+\tfrac{3\Delta}{\ln\Delta-1}}$.
\end{proof}

\section{Improved bounds for subclasses of triangle-free graphs}\label{sec:families}

\subsection{A generalized bound and an application to graphs of bounded chromatic number}

It can be noted that the value of the bound of
Theorem~\ref{th:upperbound_notriangle-new} heavily relies on
Corollary~\ref{cor:shearer}. For large values of~$\Delta$, this bound
is nearly optimal~\cite{S83}. However, directly using the slightly
stronger original bound of J.~Shearer (Theorem~\ref{thm:shearer}) or a
stronger bound holding for some particular class of graphs, one could
obtain a strengthened result as follows. Let $G$ be a nontrivial
connected identifiable triangle-free graph on $n$ vertices having
maximum degree~$\Delta$. Suppose each subgraph $H$ of $G$ has an
independent set of size at least $f(\Delta)|V(H)|$. Let
$f'(\Delta)=\min\left\{\tfrac{1}{3},f(\Delta)\right\}$. Then, the value
$\tfrac{\ln\Delta -1}{\Delta}$ in Lemma~\ref{lemm:goodIS} can be
replaced by $f'(\Delta)$, and the condition for
applying Case~1 in the proof of
Theorem~\ref{th:upperbound_notriangle-new} can be replaced by
$|Y|\geq\tfrac{3n}{\Delta f'(\Delta)+3}$. We
then get the following theorem:

\begin{theorem}\label{th:upperbound_notriangle-general}
  Let $G$ be a nontrivial connected identifiable triangle-free graph
  on $n$ vertices with maximum degree~$\Delta$ such that each subgraph
  $H$ of $G$ has an independent set of size at least
  $f(\Delta)|V(H)|$.  Let
  $f'(\Delta)=\min\left\{\tfrac{1}{3},f(\Delta)\right\}$. Then
  $\M(G)\leq n-\frac{n}{\Delta+\tfrac{3}{f'(\Delta)}}$.
\end{theorem}


It is an easy observation that any $k$-colourable graph has an
independent set of size at least $\tfrac{n}{k}$, and any subgraph of a
$k$-colourable graph is $k$-colourable. Hence we can apply
Theorem~\ref{th:upperbound_notriangle-general} to $k$-colourable
triangle-free graphs. Examples of large classes of graphs with bounded
chromatic number are for example: bipartite graphs, graphs of bounded
degeneracy, graphs having no $K_\ell$-minor~\cite{K84}, or graphs of
bounded genus~\cite{H90} --- in particular, planar triangle-free
graphs are 3-colourable following Grötzsch's theorem~\cite{G59}. We
get the following corollary:

\begin{corollary}\label{th:bip-planar}
  Let $G$ be a nontrivial connected identifiable triangle-free graph
  on $n$ vertices with maximum degree~$\Delta$ and chromatic number
  $\chi(G)$. Then $\M(G)\leq
  n-\frac{n}{\Delta+3\max\{3,\chi(G)\}}$. In particular:
\begin{itemize}
\item If $G$ is bipartite or planar, $\M(G)\leq n-\frac{n}{\Delta+9}$.
\item If $G$ is $k$-degenerate, $\M(G)\leq n-\frac{n}{\Delta+3(k+1)}$.\footnote{It is a well-known fact that a $k$-degenerate graph is $(k+1)$-colourable.}
\item If $G$ has no $K_\ell$-minor, $\M(G)\leq n-\frac{n}{\Delta+3c_1(\ell)}$, where $c_1(\ell)$ depends only on $\ell$.\footnote{It was conjectured by Hadwiger that $c_1(\ell)\leq\ell-1$~\cite{H43}, which would be optimal. However it is known that $c_1(\ell)=O(\ell\sqrt{\ln(\ell)})$~\cite{K84}.}
\item If $G$ has genus $g(G)=g$, $\M(G)\leq n-\frac{n}{\Delta+3c_2(g)}$, where $c_2(g)$ depends only on $g$.\footnote{A theorem of Heawood states that $c_2(g)\leq \left\lceil\tfrac{7+\sqrt{1+48g}}{2}\right\rceil$~\cite{H90}.}
\end{itemize}
\end{corollary}

\subsection{Graphs having no false twins}
Let $G$ be a triangle-free graph without any pair of false twins. By
considering Case~1 of the proof of
Theorem~\ref{th:upperbound_notriangle-new}, we have $Y=V(G)$, which
leads to the following bound:

\begin{theorem}\label{thm:nofalsetwins}
Let $G$ be a nontrivial connected identifiable graph $G$ on $n$
vertices having maximum degree~$\Delta$ and no pair of false
twins. Then $\M(G)\leq
n-\tfrac{n}{\tfrac{3\Delta}{\ln\Delta-1}}=n-\tfrac{n}{o(\Delta)}$.
\end{theorem}

Hence any class of connected triangle-free graphs of maximum
degree~$\Delta$ having its minimum identifying code of size at least
$n-\tfrac{n}{\Theta(\Delta)}$ should contain false twins. Note that
this is the case of the complete $(\Delta-1)$-ary tree already
mentioned in the introduction (all its leaves are false twins), and of
the classes of graphs described in~\cite{F09} (which are built using
copies of small complete bipartite graphs $K_{\Delta,\Delta}$ joined
to each other, and therefore contain many false twins).

\subsection{Graphs of girth at least~5}
In this paper, we have considered triangle-free graphs, that is,
graphs of girth at least~4. It is natural to ask whether much stronger
bounds on parameter $\M$ hold for graphs of larger girth. However note
that the answer to this question is negative because of the complete
$(\Delta-1)$-ary tree on $n$ vertices $T$, which was already mentioned
earlier. This graph has infinite girth and $\M(T)=\lceil
n-\tfrac{n}{\Delta-1+1/\Delta}\rceil$~\cite{BCHL05}.

However, with an additional condition on the minimum degree of the
graph, the question was answered in the positive in~\cite{F09} and
recently in~\cite{FP11}, where the following bounds are given.

\begin{theorem}[\cite{F09}]\label{thm:mindeg2girth5}
Let $G$ be a connected identifiable graph on $n$ vertices having minimum
degree at least~2 and girth at least~5. Then $\M(G)\le\tfrac{7n}{8}+1$.
\end{theorem}

\begin{theorem}[\cite{FP11}]\label{thm:girth5}
  Let $G$ be an identifiable graph on $n$ vertices having minimum
  degree~$\delta\geq 1$ and girth at least~5. Then
  $\M(G)\leq(\tfrac{3}{2}+o_\delta(1))\frac{\ln\delta}{\delta}n$,
  where $o_\delta(1)$ is a function of $\delta$ tending to~$0$ when
  $\delta$ tends to infinity.
\end{theorem}

Note that these two bounds are much stronger than any bound of the
form $n-\tfrac{n}{\Theta(\Delta)}$, such as the one of
Conjecture~\ref{conj}. They are best possible in the sense that relaxing
either the condition on girth~5 or minimum degree~2, there are graphs
which have much larger identifying codes. If one drops the minimum
degree~2 condition, such a graph is the complete $(\Delta-1)$-ary
tree. If one drops the girth~5 condition, there are $\Delta$-regular
graphs ($\Delta\geq 2$) having girth~4 and their minimum identifying
code of size~$n-\tfrac{n}{\Theta(\Delta)}$~\cite{F09}. We would like
to refer the interested reader to~\cite{FP11}, where this question is
studied in more detail.

\subsection{Summary of all results}
We summarize the bounds discussed in this paper in Table~\ref{table}.

\begin{table}[ht!]
\centering
\footnotesize{
\begin{tabular}{lcc}
\hline
& &\\[-0.2cm]
Graph class & Upper bound on $\M$ & Reference\\[0.1cm]
\hline\hline
& &\\[-0.2cm]
Triangle-free & $n-\frac{n}{\Delta+\tfrac{3\Delta}{\ln\Delta-1}}$ & Theorem~\ref{th:upperbound_notriangle-new}\\[0.5cm]
Bipartite & $n-\frac{n}{\Delta+9}$ & Corollary~\ref{th:bip-planar}\\[0.4cm]
Planar triangle-free & $n-\frac{n}{\Delta+9}$ & Corollary~\ref{th:bip-planar}\\[0.4cm]
Triangle-free without false twins & $n-\tfrac{n}{\tfrac{3\Delta}{\ln\Delta-1}}$ & Theorem~\ref{thm:nofalsetwins}\\[0.5cm]
Minimum degree~2, girth at least~5 & $\tfrac{7n}{8}+1$ & Theorem~\ref{thm:mindeg2girth5}~\cite{F09}\\[0.4cm]
Minimum degree~$\delta$, girth at least~5 & $\left(\tfrac{3}{2}+o_\delta(1)\right)\frac{\ln\delta}{\delta}n$ & Theorem~\ref{thm:girth5}~\cite{FP11}\\[0.2cm]
\hline
\end{tabular}}
\caption{Upper bounds in subclasses of connected identifiable graphs on $n$ vertices with maximum degree~$\Delta$}\label{table}
\end{table}



\section{On the complexity of finding a small identifying code}\label{sec:rem}

We note that our proofs provide a polynomial-time algorithm to compute
the identifying codes of
Theorem~\ref{th:upperbound_notriangle-new}. Indeed, their
constructions are based on the codes computed in
Lemmas~\ref{lemma:LRdeg2}, and~\ref{lemma:LR}, and the independent set
of Lemma~\ref{lemm:goodIS} for the first code, and on the construction
of Proposition~\ref{lemm:falsetwincode} for the second code. All these
constructions are described in the corresponding proofs and can be
done in polynomial time. Let us give an explicit complexity bound.

We observe that the running time of the constructions is at most of
the order $O(n^2\ln n)$. Indeed, the most difficult step is to compute
and compare the neighbourhoods of the vertices in order to build the
false twin equivalence classes in the proof of
Theorem~\ref{th:upperbound_notriangle-new}. To do this one can
represent each neighbourhood as a binary word of length $n$. Bitwise
comparing two of them requires $O(n)$ operations, hence a classical
sorting algorithm can sort them all in time $O(n^2\ln n)$. Comparing
them takes $O(n^2)$ time. Moreover, the construction of the
independent set of Lemma~\ref{lemm:goodIS} is based on
Theorem~\ref{thm:shearer} given in~\cite{S83}. There, the author
gives a randomized linear-time algorithm for computing the independent
set. Note that the random (constant-time) step of this algorithm can
be turned into a deterministic linear-time computation, which leads to
an $O(n^2)$ algorithm. All other steps and constructions can also be
done in time $O(n^2)$. Hence, we have the following theorem.

\begin{theorem}
  Let $G$ be a connected identifiable triangle-free graph
  on $n$ vertices with maximum degree~$\Delta\geq 3$. Then, an identifying
  code of $G$ having cardinality at most
  $n-\tfrac{n}{\Delta+\tfrac{3\Delta}{\ln\Delta-1}}$ can be computed
  in time $O(n^2\ln n)$.
\end{theorem}


\end{document}